\documentclass{patmorin}
\usepackage{pat,graphicx,amsopn}
\newcommand{\VG}{\mathit{VG}}
\newcommand{\EVG}{\mathit{EVG}}
\newcommand{\VSP}{\mathit{VSP}}
\newcommand{\Oe}{O_\epsilon}
\DeclareMathOperator{\start}{start}

\title{\MakeUppercase{Planar Visibility: Testing and Counting}}
\author{Joachim Gudmundsson\thanks{\affil{NICTA},
\email{joachim.gudmundsson@nicta.com.au}}\,
       and Pat Morin\thanks{\affil{Carleton University},
\email{morin@scs.carleton.ca}}}

\begin{document}
\maketitle
\begin{abstract}
In this paper we consider query versions of visibility testing and
visibility counting. Let $S$ be a set of $n$ disjoint line segments
in $\R^2$ and let $s$ be an element of $S$. Visibility testing is
to preprocess $S$ so that we can quickly determine if $s$ is visible
from a query point~$q$. Visibility counting involves preprocessing $S$
so that one can quickly estimate the number of segments in $S$ visible
from a query point $q$.

We present several data structures for the two query problems. The
structures build upon a result by O'Rourke and Suri (1984) who showed that
the subset, $V_S(s)$, of $\R^2$ that is weakly visible from a segment $s$
can be represented as the union of a set, $C_S(s)$, of $O(n^2)$ triangles,
even though the complexity of $V_S(s)$ can be $\Omega(n^4)$. We define
a variant of their covering, give efficient output-sensitive algorithms
for computing it, and prove additional properties needed to obtain
approximation bounds. Some of our bounds rely on a new combinatorial
result that relates the number of segments of $S$ visible from a point $p$
to the number of triangles in $\bigcup_{s\in S} C_S(s)$ that contain $p$.
\end{abstract}

\section{Introduction}

Let $S$ be a set of $n$ closed line segments whose interiors are pairwise
disjoint.  Two points $p,q\in\R^2$ are (mutually) \emph{visible} with
respect to $S$ if the open line segment $pq$ does not intersect any
element of $S$.  A segment $s\in S$ is \emph{visible} (with respect to
$S$) from a point $p$ if there exists a point $q\in s$ such that $p$
and $q$ are visible.  If two objects (points, segments) $A$ and $B$ are
visible (with respect to $S$), then we say that $A$ and $B$ \emph{see}
each other (w.r.t. $S$).  In this paper we consider the following two
problems:

\begin{prb}[Visibility testing]
  Given a query point $p$ and a segment $s\in S$, determine if $p$
  sees $s$.
\end{prb}

\begin{prb}[Visibility counting]
  Given a query point $p$, report the number of segments of $S$ visible
  from $p$.
\end{prb}

For a point $p\in\R^2$, the \emph{visibility region} or \emph{visibility
polygon} of $p$ (w.r.t. $S$) is defined as (see \figref{fig:VizRegion}.a):
\[
   V_S(p)=\{q\in\R^2:\mbox{$p$ and $q$ are visible (w.r.t. $S$)}\}
      \enspace .
\]
The visibility region of a point is star-shaped, has $p$ in its
kernel, and has size $O(n)$. It can be computed in $O(n\log n)$ time
by sorting the endpoints of segments in $S$ radially around $p$ and
then processing these in order using a binary search tree that orders
segments by the order of their intersections with a ray emanating from $p$
\cite{a85,so84}. (Equivalently, one can compute the lower-envelope of $S$
in the polar coordinate system whose origin is $p$.)  Because $V_S(p)$
is star-shaped with $p$ in its kernel it is easy to determine if a query
point $q$ is contained in $V_S(p)$ in $O(\log n)$ time using binary
search. In this way, one can consider $V_S(p)$ as an $O(n)$ sized data
structure that can test, in $O(\log n)$ time, if a query point $q$
sees $p$.

For a segment $s\in S$, the \emph{visibility region} of $s$ (with respect
to $S$)
\[
   V_S(s) = \bigcup_{q\in s} V_S(q)
          = \{p\in\R^2:\mbox{$s$ and $p$ are visible (w.r.t. $S$)} \}
\]
is the set of points in $\R^2$ that see (at least some of) $s$, see
\figref{fig:VizRegion}.b.  Unlike the visibility region of a point,
the visibility region of a segment is a complicated structure.
For a segment $s$, $V_S(s)$ can have combinatorial complexity
$\Omega(n^4)$ and $\R^2\setminus V_S(s)$ can have $\Omega(n^4)$
connected components \cite[Figure~8.13]{o87}\cite[Lemma~12]{fhjmz08},
see also \figref{quartic}.

More troublesome than the worst-case complexity of $V_S(s)$ is that there
exist sets $S$ of $n$ line segments where, for most of the elements
$s\in S$, the complexity of $V_S(s)$ is $\Omega(n^2)$.  Therefore,
explicitly computing $V_S(s)$ and preprocessing it for point location
does not yield a particularly space-efficient data structure for testing
if a query point $p$ sees $s$, even if $s$ is a ``typical'' (as opposed
to worst-case) element of $S$.

In this paper we propose efficient data structures that use an old
result of Suri and O'Rourke \cite{so84} which shows that $V_S(s)$ can
be represented as a set of $O(n^2)$ triangles whose union is $V_S(s)$.
We define a variant of their covering, give efficient algorithms for
computing it, and prove additional properties of the covering. In
particular, we define a covering $C_S(s)$ of $V_S(s)$ by triangles.
We prove that for a randomly chosen $s\in S$, the expected size of
$C_S(s)$ is $O(n)$.  This, of course, implies that $|\bigcup_{s\in S}
C_S(s)|=O(n^2)$.  Additionally, if we define $C(S)=\bigcup_{s\in S}
C_S(s)$, then we prove that the number of triangles of $C(S)$ containing
any point $p$ is a $2$-approximation to the number of segments of $S$
visible from $p$.

Applications of these results include efficient data structures for
testing if a query point is contained in $V_S(s)$ as well as efficient
data structures for estimating the number of points of $S$ visible
from a query point.  In order to express our results more precisely,
we need some further definitions.

\begin{figure}
  \begin{center}
    \includegraphics[width=11cm]{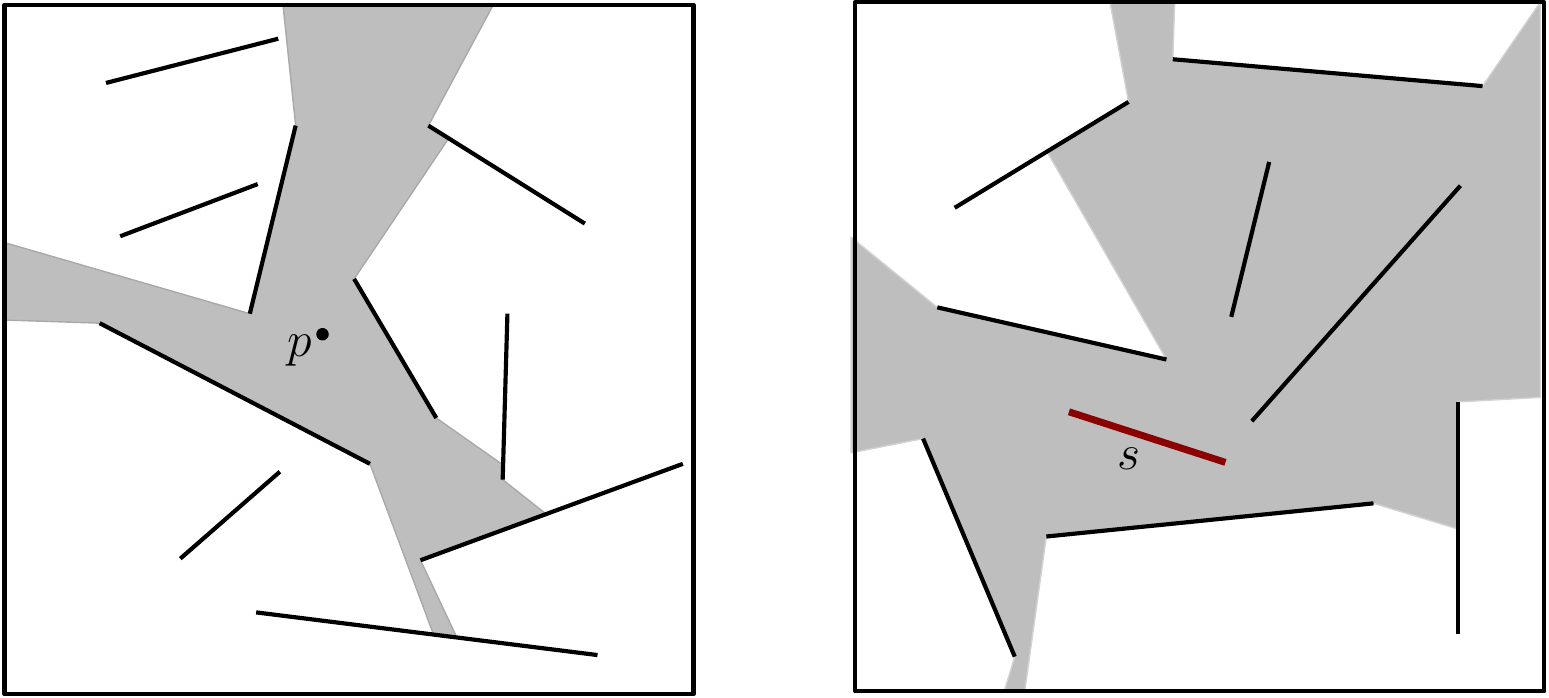}
    \caption{(a) The visibility region for a point and (b) The visibility region of a line segment.}
  \figlabel{fig:VizRegion}
  \end{center}
\end{figure}

\subsection{Visibility Graphs and Extended Visibility Graphs}

The \emph{visibility graph} $\VG(S)$ is a graph whose vertices are the
$2n$ endpoints of the segments in $S$ and in which the edge $pq$ exists
if and only if the open line segment with endpoints $p$ and $q$ does
not intersect any (closed) segment in $S$. (see \figref{vg}.a).
It is well-known that the number of edges $m$ of $\VG(S)$ is in $O(n^2)$.
Ghosh and Mount~\cite{gm91} give an optimal $O(n\log n+ m)$
time algorithm to compute the visibility graph of a set of $n$ disjoint
line segments.  Here, and throughout the remainder of the paper, $m=m(S)$
is the number of edges of $\VG(S)$.

\begin{figure}
  \begin{center}
    \includegraphics{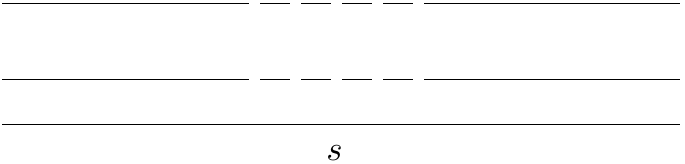}
  \end{center}
  \caption{An example of a set $S$ where $V_S(s)$ has complexity
   $\Omega(n^4)$. The $O(n)$ segments in the center define $\Omega(n^2)$
   visibility graph edges whose extensions intersect in $\Omega(n^4)$ points.}
  \figlabel{quartic}
\end{figure}

Assume, w.l.o.g., that no segment in $S$ is vertical, so
we can say that a point $p$ is \emph{above} a segment $s\in S$ if $p$ is
above the line that contains $s$.  Assume, furthermore, that $S$ contains
four segments that define a rectangle that contains all the elements of $S$
in its interior.
The first assumption can be ensured by performing a
symbolic rotation of $S$.  The second assumption is only used to ensure
that all visibility regions that we discuss are bounded.

The \emph{extended visibility graph} $\EVG(S)$ is obtained by adding
$2m$ edges and at most $2m$ vertices to $\VG(S)$ as follows (see \figref{vg}.b):
For each (directed) edge $uv$ in $\VG(S)$, extend a
segment $e_{uv}$ from $v$ in the direction $\overrightarrow{uv}$ until it
intersects an element of $S$ at some point $w$.  If not already present,
then add the vertex $w$ to $\EVG(S)$ and add the edge $vw$ to $\EVG(S)$.
The extended visibility graph can be computed in $O(n\log n + m)$ time
using the visibility graph algorithm by Ghosh and Mount~\cite{gm91}.

The union of the edges of $\EVG(S)$ and the segments in $S$ form a
1-dimensional set whose removal disconnects $\R^2$ into a set of
2-dimensional regions.  This set of 2-d regions is known as the
\emph{visibility space partition}, $\VSP(S)$ of $S$.  The regions
of $\VSP(S)$ are important because for any region $R\in\VSP(S)$
and for any $p,q\in R$ the set of segments of $S$ visible from $p$ is equal to
the set of segments of $S$ visible from $q$.  The region of $\VSP(S)$ that
contains $p$ determines all the combinatorial information about $V_S(p)$.

Note that $\VSP(S)$ is defined by $O(n^2)$ lines, rays, and segments
and therefore has worst-case complexity $O(n^4)$.

\begin{figure}
  \begin{center}
    \begin{tabular}{cc}
    \includegraphics[width=5cm]{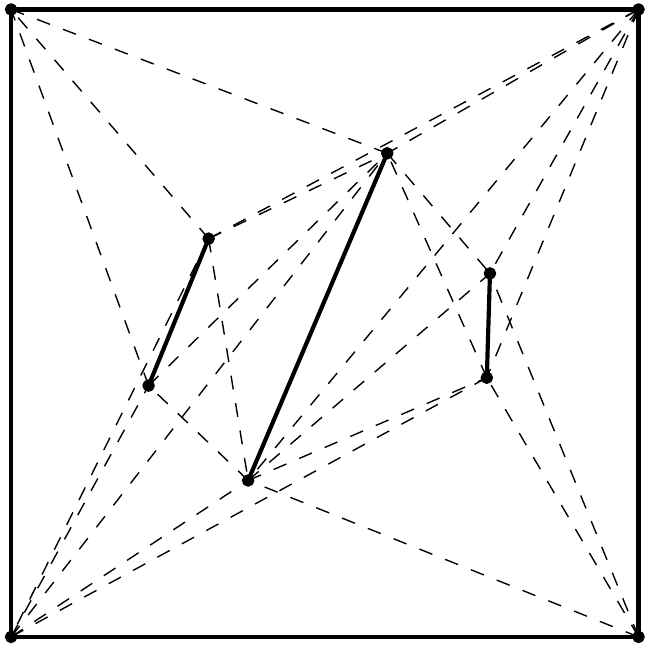} & \includegraphics[width=5cm]{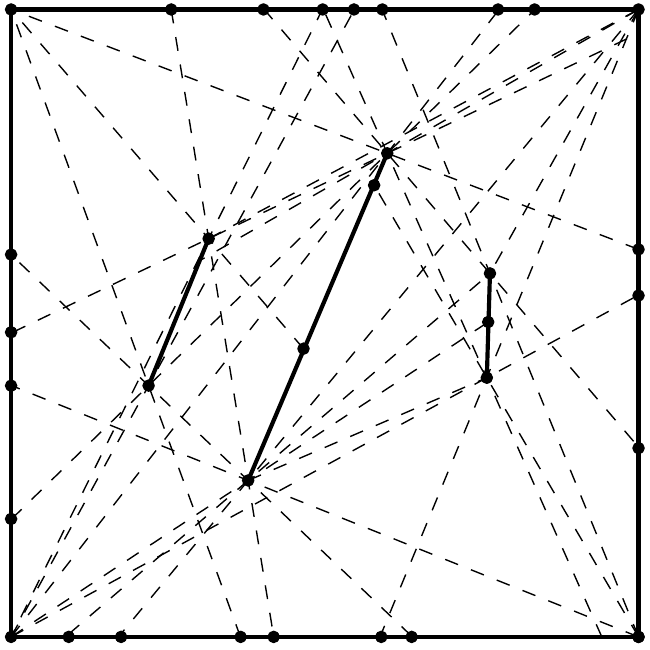} \\
    (a) & (b)
    \end{tabular}
  \end{center}
  \caption{The visibility graph and the extended visibility graph of a set
       of seven line segments. (Segments are bold, graph edges are dashed.)}
  \figlabel{vg}
\end{figure}

\subsection{Previous Work}

There is a plethora of work on visibility in the plane.  This section
discusses only some of the work most relevant to the current paper.

The visibility space partition is bounded by a subset of the $O(n^2)$ lines
induced by pairs of endpoints in $S$. The  $\VSP(S)$ has complexity $O(m^2)$
where $m$ is the number of edges in $\VG(S)$ and can be computed in $O(m^2)$
time after constructing $\VG(S)$ using standard algorithms.

By preprocessing $\VSP(S)$ with a point location structure and augmenting
the regions of $S$ with appropriate information, one obtains an $O(m^2)$
size data structure that can answer visibility testing queries and visibility
counting queries in $O(\log n)$ time.


If the segments of $S$ are the edges of a simple polygon then Bose \etal\
\cite{blm02} and Guibas \etal\ \cite{gmr97} show that the complexity of
$\VSP(S)$ is only $O(n^3)$.  In this case, this immediately solves the
two problems using a structure of size $O(n^3)$.  Aronov \etal\
\cite{agtz02} give a data structure that reduces the space to $O(n^2)$
but increases the $O(\log n)$ query time term to $O(\log^2 n)$, again
for the case where segments of $S$ are the edges of a simple polygon.

Pocchiola and Vegter~\cite{pv96} give an $O(m)$ space data structure,
the \emph{visibility complex}, that can compute the visibility
polygon $V_S(p)$ from any query point $p$ in $O(m_p \log n)$ time, where
$m_p$ is the complexity of $V_S(p)$. When the segments of
$S$ define a polygon with $h$ holes then Zarei and Ghodsi~\cite{zg-ecqpv-05}
give an $O(n^3)$ space data structure that can compute the visibility
polygon $V_S(p)$ in $O(m_p \log n)$ time and the query time of their
structure is $O(\min\{h, m_p\} \log n+m_p)$, which improves the query
time of Pocchiola and Vegter when $h \ll m_p$.

Motivated by the computer graphics problem of estimating \emph{a priori}
the savings to be had by applying a visibility culling algorithm,
Fischer \etal\ \cite{fhjmz08,fhjmz09} give approximation algorithms for
Problem~2.  They present two approximation data structures for visibility
counting. One structure uses a $(r/m)$-cutting \cite[Section~4.5]{m02}
of the $\EVG(S)$ to obtain a data structure of size $O((m/r)^2)$ that answers
queries in $O(\log n)$ time and approximates the visibility count up to
an absolute error of $r$.  Another structure uses random sampling to
obtain a data structure of size $(m^2\log^{O(1)} n)/\ell$, that has query
time $\ell\log^{O(1)} n$, and that approximates the visibility count up to
an absolute error of $\delta n$ for any constant $\delta > 0$.  (Note that
$\delta$ affects the leading constants of both the query time and space
requirements.)

\subsection{New Results}

In the current paper we revisit O'Rourke and Suri's proof that, for any
$s\in S$, there exists a set $C_S(s)$ of $O(m_s)$ triangles whose union
is $V_S(s)$, where $m_s$ is the number of edges of $\EVG(S)$ incident
on $s$.  We show that this covering has the additional property that if
we take the $O(m)$ size set $C(S)=\bigcup_{s\in S}C_S(s)$ of triangles,
then the number of triangles containing any point $p\in\R^2$ is a
$2$-approximation to the number of segments of $s$ that are visible
from~$p$.\footnote{In fact, O'Rourke and Suri's covering is a
$3$-approximation.  The slightly modified version we describe in this
paper is a $2$-approximation.}

These triangle-covering results have several applications that are
obtained by storing the resulting triangles in a layered partition tree.
Here, and throughout the remainder of the paper, $\epsilon > 0$ is a
constant that can be made arbitrarily small. To reduce clutter, we use the
notation $\Oe(f(n))=O(f(n)n^{\epsilon})$.

\subsubsection{Visibility testing}

By storing the elements of $C_S(s)$ in a partition tree, we obtain, for
any $k$ with $m_s\le k\le m_s^2$, an $O(k)$ space data structure that can
test, in $\Oe(m_s/\sqrt{k})$ time, if a query point $p$ is contained in
$V_S(s)$.  Barring a major breakthrough on Hopcroft's Problem \cite{e96},
this result is likely only a factor of $O(n^\epsilon)$ from the optimal.
See Section~\ref{sec:VizTesting}.

For comparison, the best previously described structure
for this problem, as used within the results of Fischer \etal\
\cite{fhjmz08,fhjmz09}, has size $O(m_{s}^2/\ell)$ and answers queries
in $O(\ell\log n)$ time, where $\ell \ge 1$ is a space/time tradeoff
parameter of the data structure.  Taking $\ell=\sqrt{n}$ yields a space
of $O(m_s^{3/2})$ and a query time of $O(\sqrt{m_s}\log n)$.  On the other
hand, taking $k=m_s^{3/2}$ in our data structure yields an $O(m_s^{3/2})$
space data structure with query time $\Oe(m_s^{1/4})$.

\subsubsection{Visibility Counting --- Relative Approximation}

By putting all the triangles of $C(S)$ into a partition tree, we obtain
a data structure that can $2$-approximate the number of segments of $S$
visible from any query point.  For any $k$ with $m\le k\le m^2$, this
structure has size $O(k)$ and answers queries in time $O(m/\sqrt{k})$.  The
structure returns a visibility count $m'_p$ that satisfies $m_p \le m'_p\le
2m_p$. See Section~\ref{sec:VizCountingRel}.

\subsubsection{Visibility Counting --- Absolute Approximation}

Using a selective random sampling of the segments in $S$, we obtain a data
structure of size $\Oe((cm/n)(cn)^\alpha)=\Oe(n^{1+\alpha})$
that approximates the number of segments of $S$ visible from any query point in
time $\Oe(c(m/n)^{(1/2)(1-\alpha)})=\Oe(cn^{(1/2)(1-\alpha)})$,
for any given constants $c, \delta >0$ and $0\leq \alpha \leq 1$.
With probability at least $1-n^{\Omega(\delta^2 cn/m_p)}$, the structure returns
a value $m''_p$ such that $m_p - c/n-\delta n \le m''_p\le m_p+\delta n$.
This data structure is described in \secref{VizCountingAbs1}.

Using random sampling in a different manner, we obtain a space versus query
time tradeoff. For any $k$ with $m/n \le k \le (m/n)^2$, we obtain a
structure of size $\Oe(k)$ and query time $\Oe(m/(n{\sqrt k}))$.  This
structure returns a visibility count $m''_p$ that satisfies $m_p-\delta n
\le m''_p \le 2m_p + \delta n$. The details can be found in Section~\ref{sec:VizCountingAbs2}.

\begin{figure} [tbh]
  \begin{center}
    \begin{tabular}{ccc}
      \includegraphics[width=4cm]{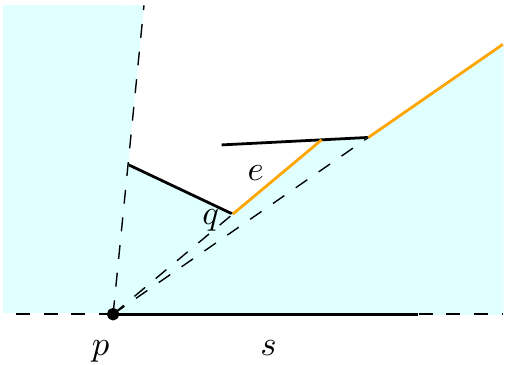} &
      \includegraphics[width=4cm]{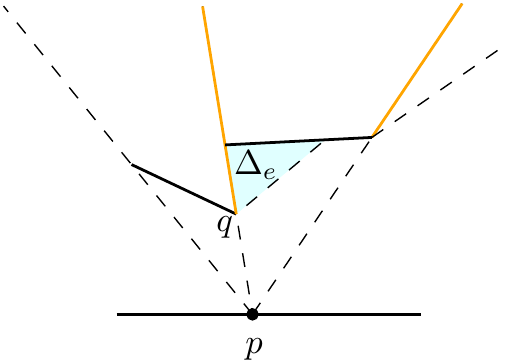} &
      \includegraphics[width=4cm]{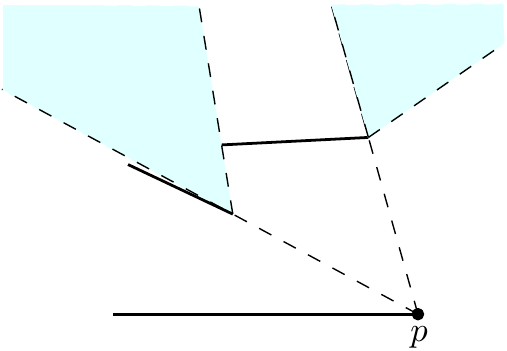} \\
      (a) & (b) & (c)
    \end{tabular}
  \end{center}
  \caption{The algorithm for covering $V^+_S(s)$ with triangles processes
           the events at (a)~$p_1$, (b)~$p_2$ and (c)~$p_3$.
           Active edges are shown in orange and triangles in the covering
           are shown at the time they are added to the covering.}
  \figlabel{alg}
\end{figure}

These results compare favourably with those of Fischer \etal\
\cite{fhjmz08,fhjmz09}.  Their cutting-based data structure, with
parameter $r=\delta n$, gives an absolute error of $\delta n$, uses
space $O((m/n)^2)$ and has a query time of $O(\log n)$.  Their random
sampling-based data structure, with parameter $\ell=\sqrt{n}$, gives
a data structure of size $(m^2\log^{O(1)} n)/\sqrt{n}$ with query time
$\sqrt{n}\log^{O(1)} n$.

The remainder of the paper is organized as follows:  \Secref{covering}
proves results on covering visibility regions with triangles.
\Secref{applications} applies these results to obtain new results on
visibility testing and counting. \Secref{conclusions} summarizes and
concludes with open problems.

\section{Covering $V_S(s)$}
\seclabel{covering}

In this section we give an algorithm for covering the visibility region
$V_S(s)$ with a set $C_S(s)$ of triangles.  The resulting covering is
similar to the covering given by Suri and O'Rourke \cite{so84}, the main
difference being around the triangles adjacent to the endpoints of $s$.
However, our exposition, and our algorithm for computing $C(S)$ are more
$s$-centric.  This leads to efficient output-sensitive algorithms for
constructing $C_S(s)$, rather than the worst-case optimal $O(n^2)$
algorithm obtained by Suri and O'Rourke \cite{so84}.
The number of triangles used in $C_S(s)$ is bounded by $O(m_s)$ where
$m_s$ is the number of edges of $\EVG(S)$ that are incident to $s$.

We will show how to cover the portion $V^+_S(s)\subseteq V_S(s)$ in the
halfplane bounded from below by the supporting line of $s$ with a set
$C^+_S(s)$ of triangles.  The complementary part $V^-_S(s)=V_S(s)\setminus
V^+_S(s)$ can be covered with a set $C^-_S(s)$ using a symmetric algorithm.

The covering algorithm works by sweeping a point $p$ from left to right
along the segment~$s$.  Events in this sweep occur at the vertices
$p_1,\ldots,p_{m'_s}$ of $\VSP(S)$ incident on $s$, in their left to right
order, so that $p_1$ and $p_{m'_s}$ are the left and right endpoints,
respectively, of $s$. 

Let $e$ be an edge of $V^+_S(p)$ that is collinear with $p$ and such that the
interior of $V^+_S(p)$ is to the right of $e$. We call such an edge an
\emph{active edge} of $V^+_S(p)$.  Active edges are important because, as $p$
moves to the right, they uncover regions of $\R^2$ which may not have been
previously visible.  See \figref{alg}.a.

Let $q$ be the lower endpoint of an active edge $e$ and note that $q$
is an endpoint of some segment in $S$. Consider what happens to $e$
as the viewpoint $p$ moves left to right along $s$, but does not cross
any edge of $\EVG(S)$ collinear with $q$.  As $p$ moves left to right,
the edge $e$ remains collinear with $p$ and $q$ and sweeps over a
triangle $\Delta_e$ whose lowest vertex is $q$.  This continues until
 the point $p$ reaches an edge of $\VSP(S)$ incident on $q$.
See \figref{alg}.b.

Algorithmically, the cover $C^+_S(s)$ is constructed as follows: Initially
$p=p_1$ is the left endpoint of $S$.  We compute the visibility polygon
$V^+_S(p)$, whose boundary is a sequence of $2m_p=O(n)$ edges that alternate
between subsegments of the elements of $S$ and segments collinear with $p$
and an endpoint of an element of $S$.  This polygon can be covered in a
natural way with $m_p$ non-overlapping triangles, each of which has $p$
as a vertex (see \figref{alg}.a). These $m_p$ triangles are added to
$C^+_S(s)$.  After computing $V^+_S(p)$ we identify its active edges, and with
each active edge $e$ we store the value $\start(e)=p_1$.

Next, we sweep $p$ from left to right, pausing at the vertices
$p_2,\ldots,p_{m'_s}$ as we go.  Upon reaching a vertex $p_i$, we process
the edges of $\EVG(S)$ incident on $p_i$ one at a time.  \footnote{For
segments in sufficiently general position, $p_i$, $1<i<m_s$ will be
incident to only one edge of $\EVG(S)$, but the covering algorithm does
not require this.} Let $e'$ be an edge of $\EVG(S)$ incident on $p_i$. If
$e'$ is collinear with an active edge $e$ of $V^+_S(p)$ then we generate a
new triangle $\Delta_e$ for $C^+_S(s)$.  The lowest vertex of $\Delta_e$
is the lower endpoint $q$ of $e$. $\Delta_e$ is bounded by two lines
$\ell_1, \ell_2$, both of which contain $q$, and where $\ell_1$ contains
$p_i$ and $\ell_2$ contains $\start(e)$.  The third side of $\Delta_e$
is bounded by the segment in $S$ incident on $e$ and furthest from $p_i$.
See \figref{alg}.b.

Finally, the visibility polygon $V^+_S(p)$ is updated in the neighbourhood
of $e$, which possibly creates up to two new active edges incident to $q$.
Each new active edge $f$ is marked as active and we set $\start(f)=p_i$.
The exact nature of this update depends on the relative locations of the
two segments that define $e'$.  The three possible cases are illustrated
in \figref{cases}.

Note that an important event, but which requires no special handling, occurs
at the right endpoint of $s$ when $p=p_{m_s}$.  In this case, each active
edge of $V^+_S(p)$ generates a triangle that is added to the set $C^+_S(s)$.
See \figref{alg}.c.

We now prove the correctness, construction time and approximation bound
of the above algorithm.

\begin{lem}\lemlabel{cover}
Let $C^+_S(s)$ be the set of triangles generated by the above
algorithm.  Then $\cup C^+_S(s) = V^+_S(s)$ and $|C^+_S(s)|\le m_s$ where
$m_s$ is the number of edges of $\VSP(S)$ incident on~$s$.
\end{lem}
\begin{proof}
To prove the bound on the size, first observe that the initial visibility
polygon $V^+_S(p_1)$ has size that is bounded by the degree of $p_1$ in
$\VSP(S)$.  Furthermore, at each event point $p_i$, $i>1$, the number of
triangles added to $C^+_S(s)$ is at most the number of edges of $\VSP(S)$
incident to $p_i$.  Therefore, the total number of triangles in $C^+_S(s)$
is at most the number of edges of $\VSP(S)$ incident on $s$.

The fact that $\cup C^+_S(s)\subseteq V^+_S(s)$ follows immediately from
the easily verifiable fact that each triangle added to $C^+_S(s)$ contains
only points visible from some point on $p\in s$. In particular, for any
point $r$ in the triangle $\Delta_e$ that is added to $V^+_S(s)$ when
processing $p_i$, there is a point $q$ in the subsegment of $s$ between
$\start(e)$ and $p_i$ that sees $r$.

To prove that $C^+_S(s)$ covers $V^+_S(s)$, consider a point $r\in
V^+_S(s)$.  If $r$ is visible from $p_1$ then $r$ is contained in
one of the triangles added during the initialization of the algorithm.
Otherwise, there exists some point $p'\in s$ with minimum $x$-coordinate
such that $r$ is visible from $p'$.  It follows that $p'$ and $r$
are collinear with a vertex $q$ of some segment $s'\in S$ and that $q$
is on the segment $p'r$ (see \figref{cover-proof}.a).  Then $q$ is an
endpoint of an active edge $e$ of $V^+_S(p')$ with $\start(e)$ to the
left of $p'$.  Since every active edge eventually adds a triangle to
$C^+_S(s)$, there is some $p_i$ to the right of $p'$ that adds a triangle
$\Delta_e$ to $C^+_S(s)$ that contains $r$ (see \figref{cover-proof}.b).
Since this is true for every point $r\in V^+_S(s)$, we conclude that
$\cup C^+_S(s)\supseteq V^+_S(s)$, and hence $\cup C^+_S(s)= V^+_S(s)$.
\end{proof}

\begin{figure}
  \begin{center}
    \begin{tabular}{cc}
      \includegraphics{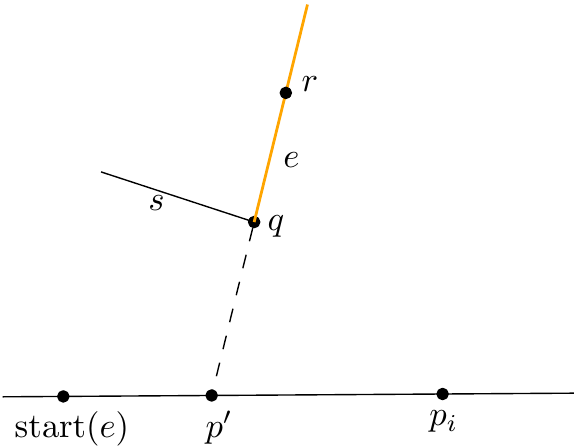} &
      \includegraphics{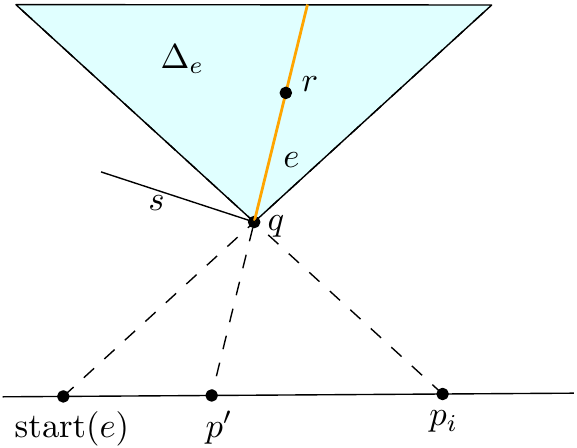} \\
      (a) & (b)
    \end{tabular}
  \end{center}
  \caption{Proving that $C^+_S(s)$ covers $V^+_S(s)$.}
  \figlabel{cover-proof}
\end{figure}

\begin{lem}\lemlabel{cover-algorithm}
  Let $S$ be a set of $n$ disjoint line segments.
  The covering $C_S(s)$ can be computed in
  \begin{enumerate}
     \item $O(m_s)$ time if we are given $\EVG(S)$ or
     \item $\Oe(n+(m_s n)^{2/3})$ otherwise.
  \end{enumerate}
\end{lem}
\begin{proof}
Part~1 of the lemma is clear. The algorithm for constructing $C_S(s)$
processes the edges of $\EVG(S)$ incident on $s$ in the order in which
they appear.  These $m_s$ edges can be easily extracted from $\EVG(S)$ in
the order in which they appear and processing each edge takes $O(1)$ time.

Part~2 of the lemma requires some use of a geometric range searching
structure for answering ray-sweeping queries.  Let $q$ and $q'$ be two
points that are visible, with $q'$ on some segment $s''\in S$.  A
\emph{ray-sweeping query} asks to determine the first endpoint of a segment
in $S$ that is intersected by $qq'$ as the point $q'$ moves towards the
left endpoint of $s''$.

A ray-sweeping query is an optimization problem.  It's corresponding
decision problem is a \emph{triangle interference query}, which asks to
determine if a query triangle $\Delta$ with vertices $q$, $q'$ and $q''\in
s''$ intersects any segment of $S$.  Because $q$ and $q'$ are visible and
$q'$ and $q''$ are both on $s''$, it is not hard to see that if $\Delta$
does intersect some segment in $S$, then $\Delta$ contains an endpoint of a
segment in $S$.  That is, a triangle interference query can be solved using
a triangular range searching structure built on the endpoints of segments
in $S$.

Triangular range searching is a well studied problem, and a number of
solutions exist that, for any $k$ with $n\le k\le n^2$, give $\Oe(k)$
space structures with $\Oe(n/\sqrt{k})$ query time \cite[Section~4]{ae99}.
Using one of these structures and applying Chan's randomized optimization
technique \cite[Theorem~3.2]{c99} yields a data structure for ray-sweeping
queries with the same preprocessing, space, and query time bounds.

To construct $C_S(s)$ we use essentially the same sweeping algorithm used
to define $C_S(s)$ except that ray-sweeping queries are used to compute
the algorithm's events on the fly.  The algorithm uses a priority queue
$Q$ to order and process these event points in left to right order.
To initialize the algorithm, we construct the visibility polygon
$V_S(p=p_1)$ in $O(n\log n)$ time using a radial sweep \cite{a85,so84}.
Next, each active edge of $V_S(p)$ is identified and processed.

Anytime (during initialization or later) that an active edge $e=qq'$
is created, the algorithm performs a ray-sweeping query with the
segment $qq'$ and a ray sweeping query with the edge $q'p$ (see
\figref{sweeping}).  The results of these two queries determine an event
point $p'\in s$ to the right of $p$ at which time the edge $e$ contributes
a triangle $\Delta_e$ to $s$. This event point $p'$ is enqueued in $Q$.
It is not hard to verify that this algorithm computes the same set of
triangles $C_S(s)$ as the original algorithm and that the number of
ray-sweeping queries performed is $O(m_s)$ (at most two queries are
performed for each triangle added to $C_S(s)$).

\begin{figure}
  \begin{center}
    \begin{tabular}{cccc}
    \includegraphics{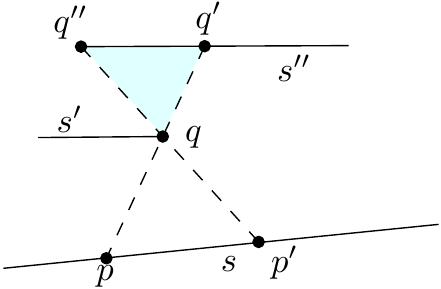} &
    \includegraphics{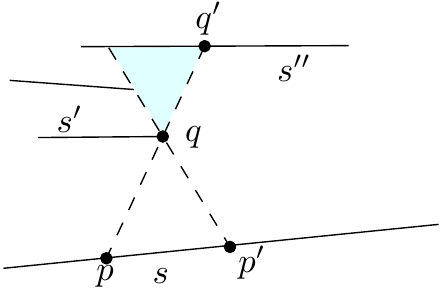} \\
    \includegraphics{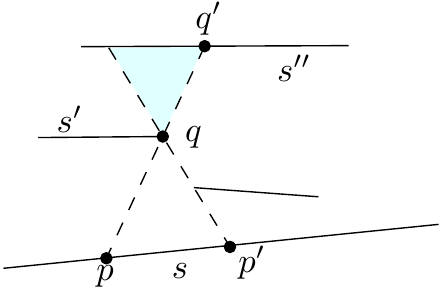} &
    \includegraphics{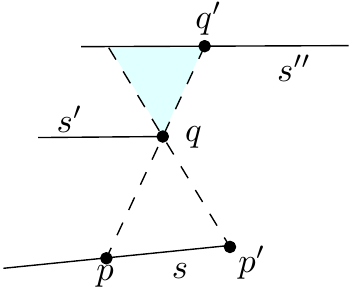} \\
    \end{tabular}
  \end{center}
  \caption{The four cases that can occur when using a ray-sweeping query
           to determine the point $p'$ at which active edge
           $e=qq'$ contributes a triangle to $C_S(s)$.}
  \figlabel{sweeping}
\end{figure}

Therefore, the algorithm can be implemented to run in $\Oe(k +
m_sn/\sqrt{k})$ for any $n\le k\le n^2$.  Given the value of $m_s$
in advance, setting $k=(m_s n)^{2/3}$ would yield the stated time
bound.  However, even without knowing $m_s$ in advance we can begin by
estimating the value of $m_s$ as $m'_s = 2$ and doubling our estimate
(and rebuilding the ray-sweeping structure) if we discover that $m_s >
m'_s$.  This doubling strategy yields the overall time bound of $\Oe(n+(m_s
n)^{2/3})$, as required.
\end{proof}

Next we show that, in a global sense, the number of triangles containing a
point $p\in\R^2$ gives a $2$-approximation to the number of segments of $S$
that are visible from $p$.

\begin{figure}[bth]
  \begin{center}
    \begin{tabular}{|c|c|c|}
      \includegraphics[width=3.8cm]{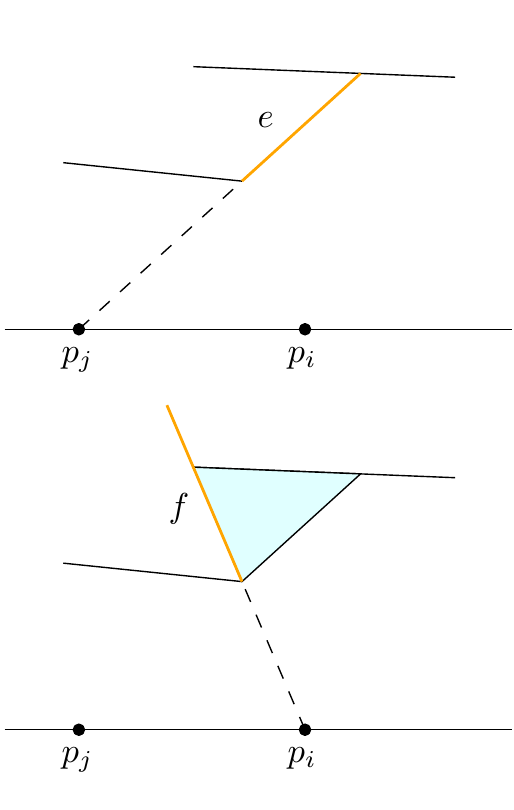} &
      \includegraphics[width=3.8cm]{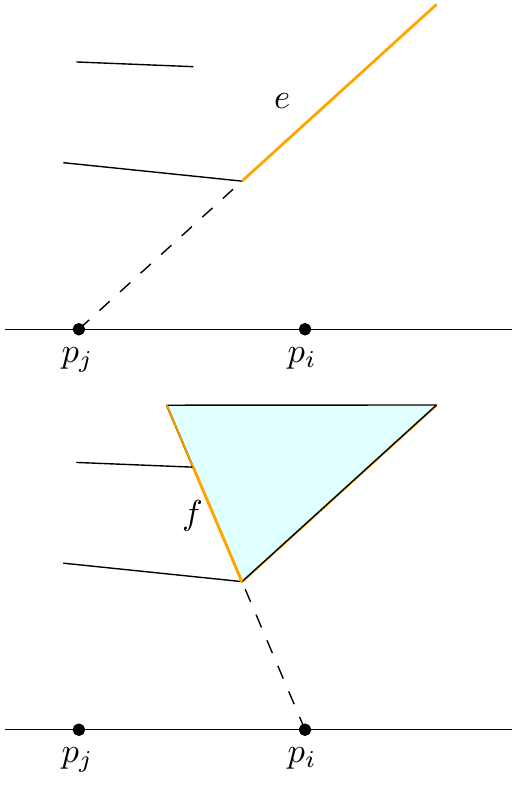} &
      \includegraphics[width=3.8cm]{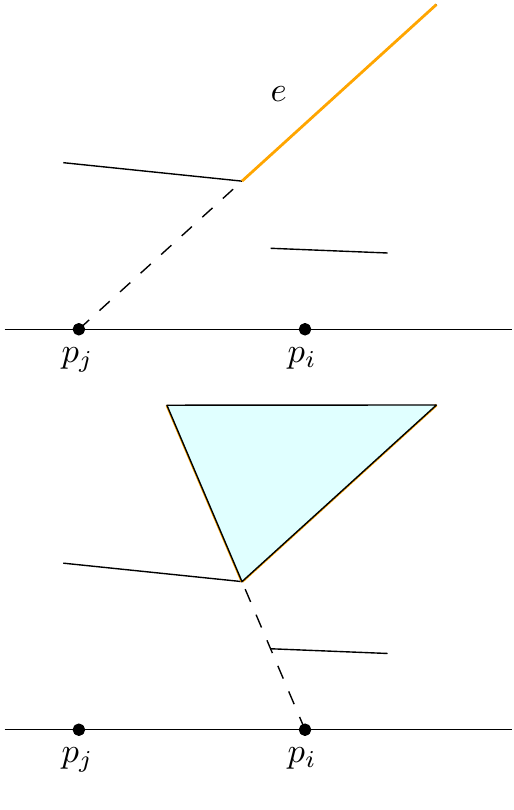} \\
      (a) & (b) & (c)
    \end{tabular}
  \end{center}
  \caption{The three cases that occur when processing an edge of $\EVG(S)$
incident on $p_i$.  Here, $\start(e)=p_j$.}
  \figlabel{cases}
\end{figure}

\begin{lem}\lemlabel{count}
Let $C(S)=\bigcup_{s\in S} C_S(s)$ and let $p$ be any point in $\R^2$ that
is not on the boundary of any triangle in $C(S)$.  If $m_p$ is the number
of segments in $S$ (partially) visible from $p$ and $m'_p$ is the number
of triangles in $C(S)$ that contain $p$, then $m_p \le m'_p \le 2m_p$.
\end{lem}
\begin{proof}
Let $C_p\subseteq C(S)$ be the set of triangles in $C(S)$ that contain
$p$, and let $S_p\subseteq S$ be the set of segments in $S$ that are
(partially) visible from $p$.  Our goal is to show that $|S_p|\le |C_p|\le
2|S_p|$. The lower bound on $m'_p=|C_p|$ is trivial: For every segment
$s\in S_p$, $V_S(s)$ contains $p$, so, by \lemref{cover}, $C_S(s)$
contributes at least one triangle to $C_p$.

To prove the upper bound, we describe a mapping $f: C_p \rightarrow S_p$
that is \emph{$2$-to-one}; for every $s\in S_p$, there exists at most two
triangles $\Delta\in C_p$ such that $f(\Delta)=s$. The existence of $f$
then proves the upper bound.

Let $\Delta\in C_p$ be some triangle that contains $p$ and suppose that
$\Delta\in C_S(s)$ for some $s\in S$ that is, without loss of generality,
below $p$.  If $\Delta$ is incident on $s$ (\figref{counting}.a), then $\Delta$ was added to
$C_S(s)$ as part of $V_S(p)$ where $p$ was the left endpoint of $s$.  In
this case, we set $f(\Delta)=s$.  Otherwise, $\Delta$ was created when
sweeping $s$ with $p$ and some active edge $e$ of $V_S(p)$ generated
$\Delta$ (\figref{counting}.b).  The vertex $q$ of $\Delta$ that is closest to $s$ is incident on a
segment $s'\in S$.  In this case $f(\Delta)=s'$.

\begin{figure}
  \begin{center}
    \begin{tabular}{cc}
      \includegraphics{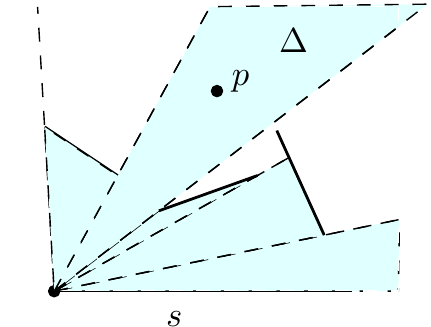} &
      \includegraphics{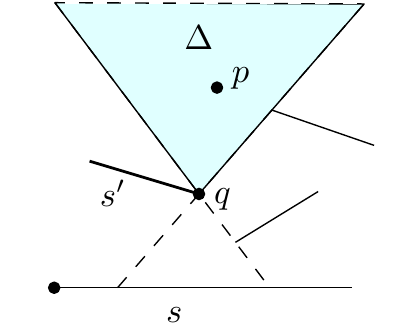} \\
      (a) $f(\Delta)=s$ & (b) $f(\Delta)=s'$
    \end{tabular}
  \end{center}
  \caption{The mapping $f$ takes $\Delta$ onto (a)~$s$ and (b)~$s'$.}
  \figlabel{counting}
\end{figure}

We now argue that $f$ is 2-to-one.  Let $s\in S$ be some segment and
suppose, without loss of generality, that $p$ is above $s$.  Consider a
triangle $\Delta \in f^{-1}(s)$ and observe that, by the definition of $f$,
$\Delta$ has a vertex that is an endpoint of $s$.

Note that there is at most one triangle in $C_S(s)\cap C_p$ that maps
to $s$, and this triangle exists precisely if $p$ is visible from the
left endpoint of $s$.  All that remains to show is that there is at
most one additional segment $s'\in S$, $s'\neq s$ such that $C_S(s')$
contains a triangle $\Delta$ with $f(\Delta)=s$.

Let $\Delta$ be such a triangle and suppose that $\Delta$ is incident to
the endpoint $q$ of $s$.  Refer to \figref{count-right}.  The triangle
$\Delta$ was generated by an active edge when processing $s'$.  In
particular, there is a subsegment $p_j p_i\subseteq s'$ such that an active
edge $e$ of $V^+_S(p)$ sweeps over $\Delta$ when $p$ travels from $p_j$ to
$p_i$. (Note, $p_j=\start(e)$.) This implies that $p_i$ and $p_j$ are below
$s$. Since $p$ travels from left to right along $e$, this implies that $q$
is the right endpoint of $s$ because, otherwise, $e$ would not be an active
edge of $V^+_S(p)$.

\begin{figure}
  \begin{center}
    \includegraphics{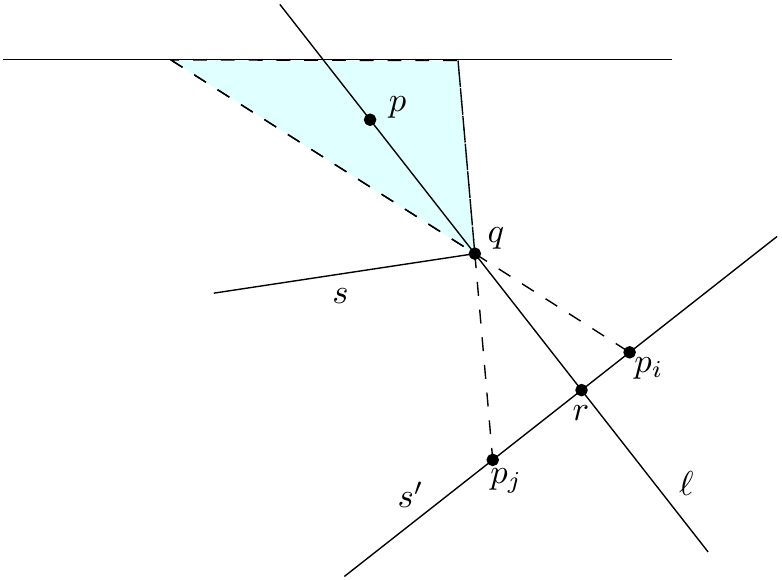}
  \end{center}
  \caption{At most one triangle in $f^{-1}(s)$ is incident to the right
           endpoint of $S$.}
  \figlabel{count-right}
\end{figure}

Thus far, we have established that at most one triangle in $f^{-1}(s)$ is
incident to the left endpoint of $s$.  To see that at most one triangle
($\Delta$, discussed above) is incident to the right endpoint of $s$,
suppose by way of contradiction that there are two such triangles $\Delta$
and $\Delta'$ with $\Delta\in C_S(s')$ and $\Delta'\in C_S(s'')$.
Consider the line $\ell$ through $p$ and $q$.  Observe that $\ell$
intersects both $s$ and $s'$, in two points $r$ and $r'$, respectively. But
this is not possible since then one of $r$ or $r'$ does not see the
endpoint $q$.
\end{proof}

\paragraph{Remark:}
The condition, in \lemref{count}, that $p$ is not on the boundary of
any triangle in $C_S(s)$ is unnecessary if we take a little extra care.
In particular, the mapping $f$ actually maps triangles to the endpoints of
segments.  The set $C_S(v)$ of triangles mapped to a particular endpoint
$v$ all have $v$ as a vertex and no two triangles in $C_S(v)$ share an
interior point.  This means that we can define each triangle in $C_S(v)$
to either include or exclude some of its edges or vertices so that the
triangles are disjoint but their union remains unchanged. This yields
a set of (partially open) triangles $C'(S)$ for which \lemref{count}
holds for any point $p\in\R^2$.

\section{Applications}
\seclabel{applications}

In this section, we consider applications of \lemref{cover} and
\lemref{count} to some visibility testing and counting problems. These
applications rely on data structures for \emph{triangle inclusion
counting}:  Given a set $T$ of triangles, we want to preprocess $T$
into a data structure for counting the number of triangles in $T$ that
contain a query point $p$.

The tools needed to perform these queries
are well-known, but finding the relevant structures and techniques,
and applying them correctly, can take some time.  Therefore, we review
the data structure here and point out the relevant references.

Let $\Delta$ be a triangle. Then $\Delta$ is the intersection of at most
4 halfplanes bounded by four lines $h(\Delta)=(u_1,u_2,d_1,d_2)$ where $u_1$
and $u_2$ bound $\Delta$ from below and $d_1$ and $d_2$ bound $\Delta$
from above. Given a  triangle $\Delta$ we have either $u_1=u_2$ or $d_1=d_2$.
By the standard duality mapping \cite[Section~8.2]{bcko08} the four lines
in $h(\Delta)$ map to four points $u_1^*$, $u_2^*$, $d_1^*$ and $d_2^*$.
A point $p\in\R^2$ maps to a line $p^*$ in the dual plane.  The point
$p$ is contained in $\Delta$ if and only if the line $p^*$ is above
(or on) $u_1^*$ and $u_2^*$ and below (or on) $d_1^*$ and $d_2^*$.
Let $h^*(\Delta)=(u_1^*,u_2^*,d_1^*,d_2^*)$.

A triangle inclusion counting structure for $T$ stores the 8-dimensional
point-set
\[
    h^*(T) = \{ h^*(\Delta) : \Delta\in T \} \enspace .
\]
Given a query point $p$, we want to count the number of points
$(a,b,c,d)\in h^*(T)$ that satisfy the four requirements:
\begin{enumerate}
  \item $a$ is above $p^*$, and
  \item $b$ is above $p^*$, and
  \item $c$ is below $p^*$, and
  \item $d$ is below $p^*$.
\end{enumerate}
Counting the number of points in $h^*(T)$ that satisfy any one of these
requirements is a halfplane range counting problem.  Data structures
for halfplane range counting are plentiful, and there are several
data structures known that use $\Oe(k)$ space and have query time
$\Oe(n/\sqrt{k})$ \cite[Section~4]{ae99}.  Several of these structures
(for example, Matou\v{s}ek's efficient partition trees \cite{m92}) are
hierarchical structures that are efficient and $r$-convergent (see Agarwal
and Erickson \cite[Section~5]{ae99} for definitions of hierarchical,
efficient, and $r$-convergent).  This implies \cite[Theorem~10]{ae99}
that there exists a 4-layer structure that uses $\Oe(k)$ space and
preprocessing time and that, in time $\Oe(n/\sqrt{k})$, can count the
number of elements in $h^*(T)$ that satisfy the constraints 1--4 for any
query point $p^*$.  Translating this back into primal space we obtain
the data structure we need:

\begin{thm}[\cite{ae99,m92}]\thmlabel{triangle-inclusion}
Let $T$ be a set of $n$ triangles. For any $k$ with $n\le k\le n^2$, there
exists a data structure of size $\Oe(k)$ that can be constructed in time
$\Oe(k)$ and that can count the number of triangles containing a query
point $p$ in $\Oe(n/\sqrt{k})$ time.
\end{thm}

\subsection{Visibility Testing} \label{sec:VizTesting}

Our first application follows immediately by storing the triangles of
\lemref{cover} in the data structure of \thmref{triangle-inclusion}.  This
yields our first result:

\begin{thm}\thmlabel{containment}
Let $S$ be a set of $n$ disjoint line segments and let $s\in S$
be a special segment.  For any $m_s\le k\le m_s^2$, there exists a
data structure of size $\Oe(k)$ that can test, in $\Oe(m_s/\sqrt{k})$ time,
if any query point $p$ is contained in $V_S(s)$.  The data structure can be constructed in
\begin{enumerate}
 \item $\Oe(k)$ time if we are given $\EVG(S)$ or
 \item  $\Oe(n + (m_sn)^{2/3} + k)$ time otherwise.
\end{enumerate}
\end{thm}

Next we argue that, barring a breakthrough on Hopcroft's Problem
\cite{e96}, \thmref{containment} is near-optimal.  \emph{Hopcroft's
Problem} takes as input a set $L$ of $n$ lines and a set $P$ of $n$ points
and asks if any point in $P$ is contained in any line in $L$. Currently,
the most efficient methods of solving Hopcroft's Problem have running
times in $\Omega(n^{4/3})$.  Furthermore, $\Omega(n^{4/3})$ is a lower
bound for Hopcroft's Problem in a restricted model of computation that
can model all known algorithms for the problem \cite{e96}.

Given the set $L$, we can compute the leftmost intersection point between
any pair of lines by sorting the lines by slope and checking the
intersection points between consecutive pairs of lines.  Assume, without
loss of generality, that this leftmost intersection point has
$x$-coordinate equal to~0. Using infinitesimal gaps between
segments,\footnote{The use of infinitesimals in lower bounds is justified
by Erickson's results \cite{e99b}.} we can easily construct a set of $3n+O(1)$
segments $s_0,\ldots,s_{3n+O(1)}$ such that a query point $p$ whose
$x$-coordinate is greater than 0 is visible from $s_0$ if and only if $p$
lies on one of the lines in $L$ (see \figref{hopcroft}). For a query point
$p$ with $x$-coordinate smaller than 0 we can test if $p$ is contained in
any line of $L$ in $O(\log n)$ time by storing the lines of $L$ sorted by
slope and using binary search.

\begin{figure}
  \begin{center}
    \begin{tabular}{cc}
      \includegraphics[width=5cm]{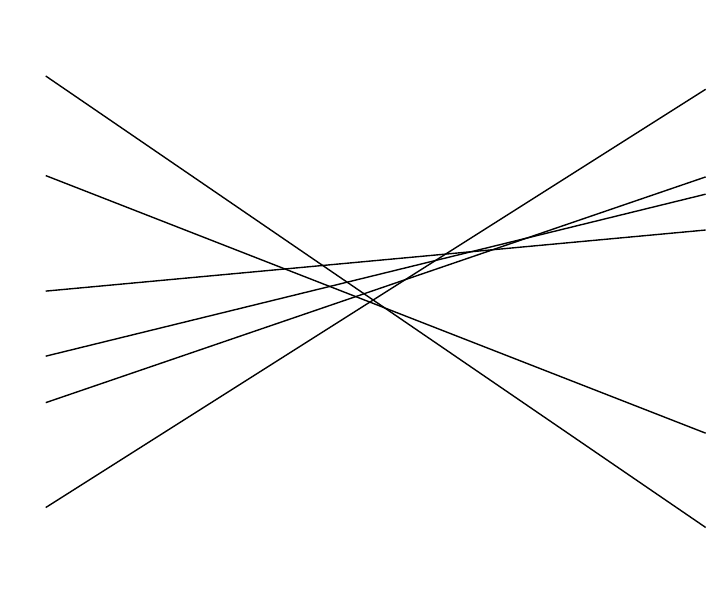} &
      \includegraphics[width=5cm]{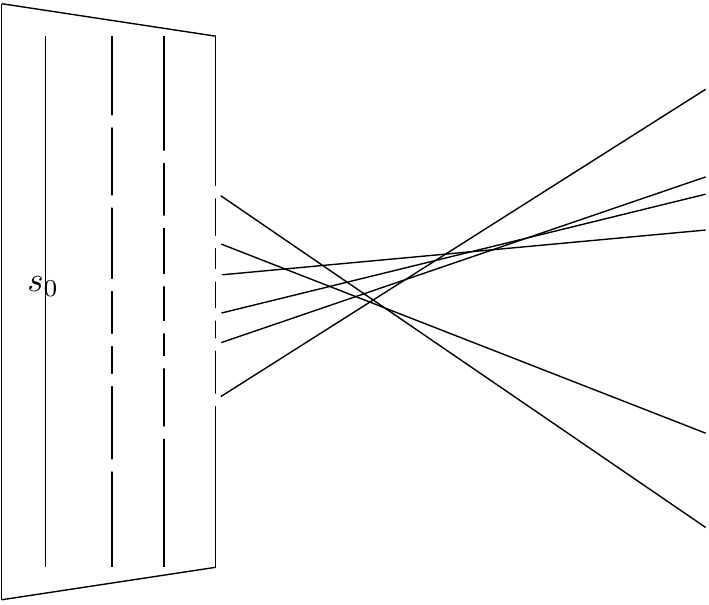} \\
      (a) & (b)
    \end{tabular}
  \end{center}
  \caption{A set $L$ of lines (a) and a set of $3n+O(1)$ segments where
           testing if a point is in $V_S(s_0)$ helps to determine if
           the point is contained in any line of $L$.}
  \figlabel{hopcroft}
\end{figure}

Therefore, by the above discussion, setting $k=n^{4/3}$ and using
\thmref{containment} we can use this data structure to solve Hopcroft's
Problem in $\Oe(n^{4/3})$ time.  Furthermore, the existence of a data
structure for testing if $V_S(s)$ contains a query point $p$ that could be
constructed in $o(n^{4/3})$ time and whose query time is $o(n^{1/3})$ would
give a $o(n^{4/3})$ time algorithm for Hopcroft's Problem.

\subsection{Visibility Counting -- Relative Approximation} \label{sec:VizCountingRel}

Next we consider Fischer \etal's problem of approximate visibility
counting \cite{fhjmz08,fhjmz09}.  We want to preprocess the segments
in $S$, so that for any query point $p$ we can quickly approximate the
number of segments in $S$ that is visible from $p$.

We begin with an easy corollary obtained by computing $C(S)$ using
\lemref{cover-algorithm} and putting all the triangles of $C(S)$ into the
data structure of \thmref{triangle-inclusion}.  The resulting structure
guarantees a relative approximation of the visibility count for all values
of $m_p$:

\begin{cor}\corlabel{relative}
  Let $S$ be a set of $n$ disjoint line segments whose visibility graph
  has $m$ edges, and let $0 < \alpha < 1$ be real valued
  parameters.  There exists a data structure $D$ that can approximate
  the number of segments of $S$ visible from any query point $p$ such that:
  \begin{enumerate}
   \item $D$ has size $\Oe(m^{1+\alpha}) = \Oe(n^{2(1+\alpha)})$,
   \item $D$ can be constructed in time $\Oe(m^{1+\alpha}) = \Oe(n^{2(1+\alpha)})$,
   \item $D$ can perform a query in $\Oe(m^{(1/2)(1-\alpha)}) =
          \Oe(n^{1-\alpha})$ time, and
   \item when querying $D$ with a point $p$ that sees $m_p$ points of
          $S$, $D$ returns a value $m'_p$ that satisfies $m_p \le m'_p
          \le 2m_p$.
  \end{enumerate}
\end{cor}

\subsection{Visibility Counting -- Absolute Approximation} 
\seclabel{VizCountingAbs}

Although \corref{relative} offers a good approximation guarantee, the
space requirement is too large.  In the worst case, when $m=\Omega(n^2)$,
a data structure of size $\omega(n^2)$ is required in order to achieve
a sublinear query time.

Fischer \etal\ \cite{fhjmz08,fhjmz09} argue that, for the computer
graphics application they consider, an absolute approximation is
sufficient.  In their application, there is a function $f(n)$ such that,
for $m_p\ll f(n)$ it is more efficient to run a visibility culling
algorithm before rendering the view from $p$ but for $f(n) \ll m_p$
it is preferable to simply send all elements of $S$ to the graphics
hardware for rendering.  For $m_p\approx f(n)$ neither strategy has a
clear advantage.  If we define $a \ll b$ as $a < b - \delta n$ then we
see that an algorithm that can approximate $m_p$ with an additive error
of at most $\delta n$ is sufficient for this application.

We present two different data structures that offer this kind of
approximation guarantee.  These two structures offer different tradeoffs
in terms of accuracy and query time.

\subsubsection{Solution 1: Sampling from $S$}
\seclabel{VizCountingAbs1}

The data structure of \thmref{containment} combined with a careful
random sampling of the elements of $S$ provides our first solution.
We create a Bernoulli sample $S''\subseteq S$ by choosing each element
of $S$ independently with probability $(c\log n)/n$, where $c \ge 1$
is a parameter of the data structure that controls the accuracy of
the approximation.  For each sample $s\in S''$, we construct the data
structure of \thmref{containment} with the value $k=m_s^{1+\alpha}$
for some parameter $0\le\alpha \le 1$ that controls the space/query-time
tradeoff.  If, during the construction of this data structure, it turns
out that $m_s > 4cn$, then discard $s$ from $S''$.  Notice that this
algorithm is effectively drawing a Bernoulli sample from the set
$S' = \{ s\in S : m_s \le 4cn \}$
and that, since $2m=\sum_{s\in S} m_s \le 4n^2$, there are at most $n/c$
elements in $S$ that are not in $S'$.
Suppose $p$ is visible from $m_p$ elements of $S$ and $m'_p$ elements of
$S'$.  Then, by the above discussion, we have  $m_p - n/c \le m'_p \le m_p$.
Let $m''_p= (n/(c\log n))\cdot |\{s\in S'': p\in V_S(s)\}|$.
The quantity $m''_p$ is an unbiased estimator of $m'_p$ and, using
Chernoff's bounds (see \appref{absolute-a}), we readily establish that
\[
   \Pr\{|m''_p - m'_p| \ge \delta n\} \le n^{-\Omega(\delta^2cn/m'_p)}
           \le n^{-\Omega(\delta^2cn/m_p)}
\]
for any $\delta > 0$. Combining this with the previous equation gives
\[
   \Pr\{|m''_p - m_p| \ge \delta n + n/c\} \le n^{-\Omega(\delta^2 cn/m_p)}.
\]
This establishes the accuracy of the data structure.  What remains is to
analyze the query time, space, and construction time.\smallbreak

\noindent\textbf{Query time.} A query computes $m''_s$ by performing
a query in each of the data structures built on the elements of $S''$.
The expected contribution of an element $s\in S'$ to the query time is
therefore
\[
    \left(\frac{c\log n}{n}\right)\cdot\Oe\left(m_s^{(1/2)(1-\alpha)}\right)
\]
since it contributes $\Oe(m_s/\sqrt{k})=\Oe(m_s^{(1/2)(1-\alpha)})$ to the
query time if it is chosen to take part in $S''$ and it contributes
nothing otherwise.  Summing this over all $s$, we get a total expected query
time of at most
\[
     \left(\frac{c\log n}{n}\right)
         \times\sum_{s\in S'}\Oe\left(m_s^{(1/2)(1-\alpha)}\right)
     = \Oe(c(m/n)^{(1/2)(1-\alpha)})
\]
where the last step follows from the fact that $f(x)=x^{(1/2)(1-\alpha)}$
is a concave function and that $\sum_{s\in S'} m_s= O(m)$.\smallbreak

\noindent\textbf{Space.}
Arguing as above, the expected amount of space that an element $s\in S'$
contributes to this data structure is
\[
    \left(\frac{c\log n}{n}\right)\cdot \Oe\left(m_s^{1+\alpha}\right).
\]
Therefore, the total expected amount of space used by the structure is
\begin{equation}
    \left(\frac{c\log n}{n}\right)\cdot \sum_{s\in S'}\Oe\left(m_s^{1+\alpha}\right)
     = \Oe((cm/n)(cn)^\alpha) \eqlabel{space-convex}
\end{equation}
where the last step follows by maximizing the sum $\sum_{s\in
S'}m_s^{1+\alpha}$ using the facts that $\sum_{s\in S'} m_s = O(m)$ and that
any individual $s\in S'$ has $m_s \le 4cn$. \smallbreak

\noindent\textbf{Preprocessing time.}
The preprocessing phase requires computing $C_S(s)$ for each sample element
$s\in S''$ and constructing a layered partition tree for the elements of
$C_S(s)$.  Constructing the partition tree takes $\Oe(m_s^{1+\alpha})$
time, so, as above, the total expected cost of constructing the partition
trees for all elements in $S''$ is $\Oe((cm/n)(cn)^\alpha)$.

Computing $C_S(s)$, using \lemref{cover-algorithm} takes
$\Oe((m_sn)^{2/3})$ time. Since $f(x) = (xn)^{2/3}$ is a concave function,
the total expected time to compute $C_S(s)$ for each $s\in S''$ is
$\Oe(cm^{2/3})$. 

\begin{thm}\thmlabel{absolute-a}
  Let $S$ be a set of $n$ disjoint line segments whose visibility
  graph has $m$ edges and let $c>1$ and $0 < \alpha < 1$ be real valued
  parameters.  There exists a data structure $D$ that can approximate
  the number of segments of $S$ visible from any query point $p$ such that:
  \begin{enumerate}
   \item $D$ has expected size $\Oe((cm/n)(cn)^{\alpha}) =
          \Oe((cn)^{1+\alpha})$,
    \item $D$ can be constructed in
          $\Oe(cm^{2/3} + (cm/n)(cn)^{\alpha}) = \Oe(cn^{4/3} + (cn)^{1+\alpha})$
          expected time,
    \item $D$ can perform a query in $\Oe(c(m/n)^{(1/2)(1-\alpha)}) =
          \Oe(cn^{(1/2)(1-\alpha)})$ expected time, and
    \item for any $\delta > 0$, when querying $D$ with a query for a
          point $p$ that sees $m_p$ points of $S$, $D$
          returns a value
          $m''_p$ that satisfies $m_p - n/c - \delta n \le m''_p \le m_p
          + \delta n$ with probability at
          least $1-n^{-\Omega(\delta^2
          cn/m_p)}$.
   \end{enumerate}
\end{thm}

\noindent\textbf{Example:} For any constant $\delta$ there exists
a $c=c(\delta)$ such that taking $\alpha = 4/3$ gives a data
structure of size $\Oe((m/n)n^{1/3})=\Oe(n^{4/3})$ with query time
$\Oe((m/n)^{1/3})=\Oe(n^{1/3})$ and the structure approximates $m_p$
for any $p$ with an absolute error of at most $\delta n$ w.h.p..

As this example shows, the structure of \thmref{absolute-a} is quite
efficient for constant values of $\delta$. Unfortunately, the theorem
becomes weaker when using subconstant values of $\delta$.
This is because, to obtain meaningful error bounds, we require
$c=\Omega(1/\delta)$ and the running time of the query algorithm grows
linearly with $c$.

\begin{figure}
  \begin{center}
    \begin{tabular}{cc}
      \includegraphics{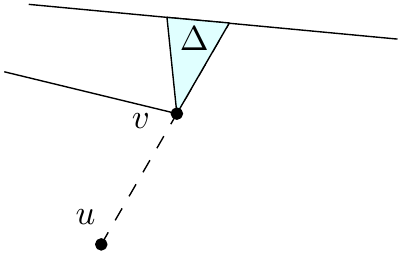} &
      \includegraphics{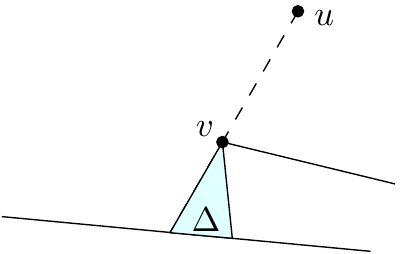} \\
    \end{tabular}
  \end{center}
  \caption{The two cases in which the ordered pair $(u,v)$ generates the
           triangle $\Delta$.}
  \figlabel{uv-sample}
\end{figure}

\subsubsection{Solution 2: Sampling $C(S)$}
\seclabel{VizCountingAbs2}

Next we consider a different data structure that is also based on
random sampling. Rather than sample segments of $S$, we instead sample
triangles of $C(S)$ and use \lemref{count} to bound the quality of the
approximation.  This results in a more efficient space/accuracy tradeoff
than that provided by \thmref{absolute-a}.  The cost of this savings
in space is that we obtain a relative approximation bound when $m_p$
is large and an absolute approximation bound when $m_p$ is small. For
the application proposed by Fischer et al. \cite{fhjmz08,fhjmz09} this
is an acceptable approximation bound.

Consider the set $C(S)$ of triangles described in \lemref{count}.
For any point $p\in\R^2$, the number, $m'_p$, of triangles in $C(S)$ that
contain $p$ is a 2-approximation of the number, $m_p$,  of segments in $S$
that are visible from $p$.  In particular
\begin{equation}
  m_p \le m'_p \le 2m_p.  \enspace . \eqlabel{mpmp}
\end{equation}

Our strategy is to approximate $m'_p$ by sampling elements of
$C(S)$.  The easiest way to proceed would be to select a Bernoulli sample
by sampling each element of $C(S)$ independently with probability $(c\log
n)/n$.  This would require enumerating the elements of
$C(S)$, of which there are $\Theta(m)$, yielding a construction time that
is $\Omega(n^2)$ in the worst case.  Instead, we use a different sampling
strategy based on the rejection method that avoids computing $C(S)$. \smallbreak

\noindent\textbf{Sampling $C(S)$.}
Our goal is to obtain a random multiset $C''(S)\subseteq C(S)$ of size
roughly $c(m/n)\log n$.  To achieve this, we repeat the following
procedure $4cn\log n$ times:  We select two points $u$ and $v$ at random,
with replacement, from the $2n$ endpoints of $S$.  Note that there are
$4n^2$ ways of doing this.  Next, using $O(1)$ ray sweeping queries,
we determine if $q$ is a vertex of some triangle $\Delta\in C(S)$ that
has an edge collinear with $\overrightarrow{uv}$ and that lies to the
left of $\overrightarrow{uv}$ (see \figref{uv-sample}).  Note that,
for any $\Delta\in C(S)$, there is exactly one pair $(u,v)$ for which
this is true.\footnote{The pair $(u,v)$ generates the triangle $\Delta$
precisely if $u$ is below $v$ and $v$ is the right endpoint of its
segment or $u$ is above $v$ and $v$ is the left endpoint of its segment.}
Therefore, if this test is affirmative then $\Delta$ is an element drawn
uniformly at random from $C(S)$ and we add it to our sample $C''(S)$.
The probability that we increase the size of $C''(S)$ this way is
$m_S/(4n^2)$, where $m_S=|C(S)|=O(m)$. \smallbreak


\noindent\textbf{Space, preprocessing time, and query time.}
To compute $C''(S)$ efficiently we use the ray-sweeping data structure
described in the proof of \lemref{cover-algorithm}.  Each sampling
step requires $O(1)$ ray-sweeping queries, which can be done in
$\Oe(n/\sqrt{\ell})$ time after $\Oe(\ell)$ preprocessing.  Thus, the
expected time required to build the ray-sweeping data structure and
perform $4cn\log n$ sampling steps is
$
    \Oe(\ell + (cn\log n)(n/\sqrt{\ell})
       = \Oe(\ell + cn^2/\sqrt{\ell})
       = \Oe(c^{2/3}n^{4/3})
$
for $\ell=c^{2/3}n^{4/3}$.

Each sampling step adds an element to $C''(S)$ with probability $m_S/(4n^2)$.
So, the number of samples in $C''(S)$ is a binomial random
variable with parameters $4cn\log n$ and $m_S/(4n^2)$ and the expected
size of $C''(S)$ is therefore $c(m_S/n)\log n$. Using Chernoff's
Bounds, we find that the probability that the size of $C''(S)$
exceeds $ac(m_S/n)\log n$ is at most $n^{-\Omega(a)}$ for any $a>1$.
This concentration result ensures that when building the data structure
of \thmref{triangle-inclusion} on the elements of $C''(S)$ the expected
size, preprocessing time, and query time of the resulting structure are
$\Oe(k)$, $\Oe(k)$, and $\Oe((m/n)/\sqrt{k})$, respectively, for any
$m/n\le k\le (m/n)^2$.

To summarize, for any $k$ with $m/n \le k \le (m/n)^2$, the above
sampling procedure runs in $\Oe(c^{2/3}n^{4/3} + k)$ expected time and
produces a data structure of $\Oe(k)$ expected size, that can answer
queries in $\Oe((m/n)/\sqrt{k})$ expected time.  All that remains is to
calibrate and check the accuracy of the results provided by the data
structure. \smallbreak

\noindent\textbf{Estimating $m'_p$}.
Recall that our goal is to estimate $m'_p$, the number of elements of $C(S)$
that contain the query point $p$, as $m'_p$ is a 2-approximation to the
number of segments of $S$ visible from $p$.  Let
\[
    m''_p = (n/(c\log n))\cdot |\{ \Delta\in C''(S) : p\in \Delta\}|. \enspace .
\]
(Note that computing $m''_p$ does not require knowing the value
$m_S=|C(S)|$.)  Each step of the sampling procedure finds an element
$\Delta\in C(S)$ such that $p\in\Delta$ with probability exactly $m'_p
/ (4n^2)$.  Since the sampling procedure runs for $4cn\log n$ steps, this
implies that the number of triangles in $C''(S)$ that contain $p$ is a
binomial random variable with parameters $4cn\log n$ and $m'_p/(4n^2)$.
Therefore,
\[
    \E[m''_p] = (n/(c\log n))(4cn\log n)(m'_p/4n^2) = m'_p \enspace .
\]
That is, $m''_p$ is an unbiased estimator of $m'_p$.  Furthermore, applying
Chernoff's bounds to the underlying binomial random variable (see
\appref{absolute-b}),
we find that
\[
  \Pr\{|m''_p - m'_p| \ge \delta n\} \le n^{-\Omega(\delta^2cn/m_p)}
\]
for any $\delta > 0$.  Combining this with \eqref{mpmp} we obtain
\[
   \Pr\{m_p-\delta n \le m''_p \le 2m_p+\delta n\}
     \ge 1-n^{-\Omega(\delta^2 cn/m_p)}  \enspace .
\]
This establishes the accuracy of the data structure and completes the proof
of our last theorem:

\begin{thm}\thmlabel{absolute-b}
  Let $S$ be a set of $n$ disjoint line segments whose visibility
  graph has $m$ edges and let $c>1$ and $0 < \alpha < 1$ be real valued
  parameters.  There exists a data structure $D$ that can approximate
  the number of segments of $S$ visible from any query point $p$ such that:
  \begin{enumerate}
   \item $D$ has expected size $\Oe((cm/n)^{1+\alpha}) =
          \Oe((cn)^{1+\alpha)})$,
   \item $D$ can be constructed in time
          $\Oe(c^{2/3}n^{4/3} + (cn)^{1+\alpha})$,
   \item $D$ can perform a query in $\Oe((cm/n)^{(1/2)(1-\alpha)}) =
          \Oe((cn)^{(1/2)(1-\alpha)})$ time, and
   \item for any $\delta > 0$, when querying $D$ with a point $p$
          that sees $m_p$ points of $S$, $D$ returns a
          value $m''_p$
          that satisfies $m_p - \delta n \le m''_p \le 2m_p + \delta n$
          with probability at least $1-n^{-\Omega(\delta^2 cn/m_p)}$.
  \end{enumerate}
\end{thm}

\noindent\textbf{Example.}  Taking $c=dn^{1/3}$, for a large constant $d$,
and $\alpha=0$, we get a data structure of size $\Oe(n^{4/3})$ that can
be constructed in time $\Oe(n^{14/9})$ and that can, in $\Oe(n^{2/3})$
time, effectively distinguish between viewpoints $p$ where $m_p \ll
n^{2/3}$ and viewpoints $p$ where $m_p \gg n^{2/3}$.

\section{Summary and Conclusions}
\seclabel{conclusions}


Many open questions remain.  The data structure for testing if a point is
in $V_S(s)$ for a segment $S$ (\thmref{containment}) is near-optimal, at
least assuming an $\Omega(n^{4/3})$ lower-bound for Hopcroft's Problem.
However, it is difficult to say if the data structures for approximate
visibility counting are close to optimal.  Our solutions reduce visibility
counting to the problem of computing the depth of a query point in an
arrangement of $O((m/n)\log n)$ (\thmref{absolute-a} and
\thmref{absolute-b}) or $O(m)$ (\corref{relative}) triangles.  In both
cases, it would be sufficient to give a relative approximation for the
depth of the query point.  Unfortunately, without some additional
assumptions (such as fatness) about the triangles, there is currently no
good solution to this problem.

The results in the current paper consider the problem of \emph{planar}
visibility counting, where $S$ is a set of disjoint line segments in
$\R^2$.  Of course, modern virtual environments are often 3-dimensional.
Many of these environments are just barely 3-dimensional in the sense
that they consist of a constant number of 2-dimensional layers that can
be handled using the data structures presented in the current paper.
However, ultimately we would like to develop data structures that store
a set $S$ of disjoint triangles in $\R^3$ and can approximately count
the number of elements of $S$ (at least partly) visible from a query
point $p\in \R^3$.

\section*{Acknowledgements}

This work was done while the second author was a visiting researcher at NICTA
and the University of Sydney.  The author is grateful for the hospitality
and funding provided by both institutions.

\bibliographystyle{plain}
\bibliography{viscover}

\appendix

\section*{Appendix}

\section{Accuracy Bound for \thmref{absolute-a}}
\applabel{absolute-a}

In this appendix, we derive the error bound of the data structure of
\thmref{absolute-a}.  To do this, we will use a version of Chernoff Bounds
for binomial random variables \cite[Appendix~A.1]{as08} which states
that, for a binomial random variable $B$ with mean $\mu$,
\begin{equation}
  \Pr\{|B-\mu| \ge \tau\mu\}
     \le \exp(-\Omega(\tau^2\mu)) \enspace .
  \eqlabel{chernoff}
\end{equation}
for any $\tau > 0$.

Let $B$ be the number of samples in $S''$ visible from $p$, let $x=m'_p$
be the number of segments in $S'$ visible from $p$, and let $t=(c\log
n)/n$.  Then $B$ is a binomial$(x,t)$ random variable with expectation
$\mu = xt$.  We have that $|m''_p-m'_p|\ge \delta n$ if and only if
$|B-\mu|\ge t\delta n$.  Taking $\tau = t\delta n/\mu$ and applying
\Eqref{chernoff} we obtain

\begin{eqnarray*}
   \Pr\{ |m''_p - m'_p| \ge \delta n \}
   & = & \Pr\{ |B-\mu| \ge t \delta  n \} \\
   & \le & \exp(-\Omega(\left(t \delta  n /\mu\right)^2 \mu)) \\
   &  = & \exp(-\Omega((t \delta  n)^2/\mu)) \\
   &  = & \exp(-\Omega((\delta n)^2 t/x)) \\
   &  = & \exp(-\Omega((\delta^2 cn\log n)/x)) \\
   &  = & n^{-\Omega(\delta^2 cn/x))} \\
   &  = & n^{-\Omega(\delta^2 cn/m'_p))} \\
   &  = & n^{-\Omega(\delta^2 cn/m_p))} \enspace ,
\end{eqnarray*}
as required.

\section{Accuracy Bound for \thmref{absolute-b}}
\applabel{absolute-b}

Let $B$ be the number of sample triangles in $C''(S)$ that contain $p$,
let $x=m'_p$ be the number of triangles in $C(S)$ that contain $p$, and
let $t=x / (4n^{2})$.  Then $B$ is a binomial$(4cn\log n,t)$ random
variable with expected value $\mu = 4tc\log n = (cx\log n) / n$.

We have that $|m''_p-m'_p|\ge \delta n$ if and only if $|B-\mu|\ge \delta c
\log n$.  Taking $\tau = (\delta c\log n)/\mu$ and applying
\Eqref{chernoff} we
obtain
\begin{eqnarray*}
  \Pr\{ |m''_p - m'_p| \ge \delta n \}
   & = & \Pr\{ |B-\mu| \ge \delta c \log n  \} \\
   & = & \exp(-\Omega(((\delta c\log n) / \mu)^2\mu)) \\
   & = & \exp(-\Omega((\delta c\log n)^2/\mu)) \\
   & = & \exp(-\Omega(\delta^2(cn\log n)/x)) \\
   & = & n^{-\Omega(\delta^2cn/x))} \\
   & = & n^{-\Omega(\delta^2cn/m'_p))} \\
   & = & n^{-\Omega(\delta^2cn/m_p))}  \enspace ,
\end{eqnarray*}
as required.

\end{document}